\documentclass{article}

\usepackage{blindtext} 

\usepackage{amsmath,amssymb,amsfonts,amsthm,enumerate}
\usepackage[sc]{mathpazo} 
\usepackage[T1]{fontenc} 
\usepackage{microtype} 
\usepackage{fullpage}
\usepackage{graphicx}
\usepackage[english]{babel} 

\usepackage[hmarginratio=1:1,top=32mm,columnsep=20pt]{geometry} 
\usepackage[hang, small,labelfont=bf,up,textfont=it,up]{caption} 
\usepackage{booktabs} 

\usepackage{lettrine} 

\usepackage{enumitem} 
\setlist[itemize]{noitemsep} 

\usepackage{abstract} 

\usepackage{titlesec} 
\renewcommand\thesection{\Roman{section}} 
\renewcommand\thesubsection{\roman{subsection}} 
\titleformat{\section}[block]{\large\scshape\centering}{\thesection.}{1em}{} 
\titleformat{\subsection}[block]{\large}{\thesubsection.}{1em}{} 

\usepackage{fancyhdr} 

\usepackage{titling} 

\usepackage{hyperref} 
\newtheorem{Theor} {Theorem}
\newtheorem{rk}{Remark}
\newtheorem{cor}{Corollary} 
\newtheorem{defi}{Definition}
\newtheorem{prop}{Proposition}

\pretitle{\begin{center}\Huge\bfseries} 
\posttitle{\end{center}} 
\title{Deterministic partial binary circulant compressed sensing matrices} 
\author{%
\textsc{Arman Arian and \"{O}zg\"{u}r Y\i lmaz} \\[1ex] 
\normalsize Department of Mathematics, \\ University of British Columbia \\ 
}
\date{} 
\renewcommand{\maketitlehookd}{%
\begin{abstract}

Compressed sensing (CS) is a signal acquisition paradigm to simultaneously acquire and reduce dimension of signals that admit sparse representations. This is achieved by collecting linear, non-adaptive measurements of a signal, which can be formalized as multiplying the signal with a ``measurement matrix". Most of matrices used in CS are random matrices as they satisfy the restricted isometry property (RIP) in an optimal regime with high probability. However, these matrices have their own caveats and for this reason, deterministic matrices have been proposed. While there is a wide classes of deterministic matrices in the literature, we propose a novel class of deterministic matrices using the Legendre symbol. This construction has a simple structure, it enjoys being a binary matrix, and having a partial circulant structure which provides a fast matrix-vector multiplication and a fast reconstruction algorithm. We will derive a bound on the sparsity level of signals that can be measured (and be reconstructed) with this class of matrices.  We perform quantization using these matrices, and we verify the performance of these matrices (and compare with other existing constructions) numerically.

Index Terms--compressed sensing, deterministic measurement matrices, binary, circulant matrices, Legendre symbol.

\end{abstract}
}

\begin{document}

\maketitle



\label{ourconst}
\section{Introduction}

In this paper, we present a novel construction for deterministic CS matrices  based on decimated Legendre sequences. As we know, Legendre sequence provides a binary sequence with $\pm 1$ entries which initially seems ideal to use in the context of CS. However, in order to be able to use these sequences as rows or columns of a measurement matrix, one has to guarantee a low maximum correlation between two such sequences. This was done first by Zhang et al. \cite{zhang} in 2002 (before the birth of CS) in the context of coding theory, and by considering \textit{decimated} Legendre sequences. The use of Legendre symbol in CS with random matrices was proposed by Bandeira et al. \cite{bandeira} in 2016. The summary of their work is given in Section \ref{sumlegendre}.   In the same year, Zhang et al. \cite{detlegendre} proposed  the use of Legendre symbol for the construction of deterministic CS matrices. In fact, their construction was based on their previous work in 2002  which was done in the context of coding theory, and  has the feature of  being binary, and having low coherence. Moreover, since any prime number can be used as the number of measurements,  the difference between size of two adjacent matrices in their construction is low compared to many other deterministic constructions. 

As outlined below, another important feature that a CS matrix can have is the \textit{circulant matrix} structure (see Section \ref{circulantmat}). The use of circulant matrices in random CS was first proposed by Bajwa et al. \cite{bajwa} in 2007. An equivalent approach was proposed by Romberg \cite{romcirculant} in 2009. In the latter approach, given a signal $x \in \mathbb{R}^n$, first a \textit{convolution} of $x$, of the form $Hx$ is considered, where $H$ can be written in the form of $$H=\frac{1}{\sqrt{n}} 	F^* \Sigma F$$ Here, $n$ is as usual the ambient dimension, $F$ is the discrete Fourier matrix, and $\Sigma$ is a diagonal matrix whose diagonal entries are complex numbers with unit norm, and random phases. Following this step, we subsample the measurements. Therefore, the measurement matrix in this approach can be written as $\Phi=R_{\Omega} H$, where $\Omega \subseteq \{1,2,\cdots,n\}$ is a set with $m$ elements, and $R_{\Omega}$ is a sampling operator that restricts the rows to a random set $\Omega$, i.e., a random choice of a set $\Omega$ among all possible $\binom{n}{m}$ such sets. Based on this idea, Li  et al. \cite{liconv} considered a measurement matrix of the form $\Phi=R_{\Omega} A$, where $A$ is a deterministic matrix. Since $R_{\Omega}$ here is a random sampling operator, the measurement matrix in their construction can not be considered as \textit{deterministic} yet. To the best of  our knowledge, the only class of \textit{circulant deterministic} matrices considered in the literature so far is the class of matrices introduced by Cui \cite{cuiconst}. In their paper, they constructed the circulant matrix $A$ by first writing it of the form $A=\frac{1}{\sqrt{n}} F^* \mbox{diag}(\sigma) F$, then, considering the sequence $\sigma$ as a \textit{decimated Legendre} sequence, and finally considering the \textit{first} $m$ rows of $A$ as the measurement matrix. They show \textit{empirically} that this matrix performs very well as a measurement matrix, but no proof was given in this regard. Other \textit{deterministic binary} matrices given in the literature in the context of CS include DeVore construction \cite{devore}, and the constructions given in \cite{li,naidu,amini}.  

To the best of our knowledge, the construction we introduce in this chapter is the \textit{first} deterministic binary circulant construction in CS which is proved to have low coherence, and hence, can be used for recovery of sparse signals. Compared to the work of Cui \cite{cuiconst}, in addition to giving a proof for why the construction can be used in CS, our construction has the advantage of having a \textit{simple}, \textit{explicit} formula for each entry of the measurement matrix itself (as opposed to its diagonalization). The circulant structure of our construction allows us to perform a fast matrix-vector multiplication, and a fast recovery algorithm. Moreover, our construction has a small  difference between the sizes of two adjacent matrices. (as we will see below, the number of measurements in our construction is chosen as $ \lceil p^{3/4} \rceil $, where $n=p$ is a prime number and is assumed to be the ambient dimension). Lastly, we will show that we can perform the one-stage recovery $\Sigma \Delta$ quantization as given in \cite{rongrong} and generalized in \cite{armanozgur} using our construction. Similar to the constructions given in \cite{cuiconst,bandeira,detlegendre}, our construction exploits Legendre symbol.  The definition and basic properties of Legendre symbol can be found in any elementary Number Theory textbook, e.g., in \cite{jones}.

\section{Compressed sensing matrices using Legendre sequence}
\label{sumlegendre}
Fix a prime number $p$, and an element $x \in \mathbb{Z}_p$. Using the Legendre symbol, we observe that the entries of any vector of the form  $v=( \Big( \frac{1+x}{p} \Big), \Big( \frac{2+x}{p} \Big), \cdots, \Big( \frac{p+x}{p} \Big) )$ are evenly distributed with $\pm 1$ entries. Hence, one can construct CS matrices using these vectors (or similar vectors). In \cite{bandeira}, Bandeira et al. constructed a class of $m \times n$ random matrices whose $(i,j)$th entry is obtained via \begin{equation} \label{legendreconst} \Phi_{i,j} :=\frac{1}{\sqrt{m}} \Big( \frac{x+m(j-1)+i}{p} \Big) \end{equation}

\noindent where $p$ is a large prime number, and $x$ is drawn uniformly from a set of the form $\{0,1,2, \cdots, 2^h-1 \}$. They show that for appropriate choices of $k, m, n, \delta$, and $h$, such matrices are RIP of order $k$ and constant $\delta$ with high probability. They also make a conjecture implying that the class of deterministic matrices obtained from fixing $x=0$ in  \eqref{legendreconst} are RIP in the \textit{optimal regime} (and hence, would break the square-root barrier). As mentioned above, one can also use the Legendre symbol to construct binary deterministic matrices. Examples of such constructions are given in  \cite{detlegendre,cuiconst}. 

Since our goal in this chapter is to consider a class of deterministic binary \textit{circulant} matrices, we next summarize the definition and main features of circulant matrices in the context of CS.   

\section{Circulant matrices}
\label{circulantmat}

The use of (random) circulant matrices in CS  has been originally suggested by Bajwa et al. \cite{bajwa} in 2007, followed by the results by Romberg \cite{romcirculant, rauhutcirculant} in 2009. As stated in \cite{rauhutcirculant}, there is an advantage for using circulant matrices compared to Gaussian and Bernoulli random matrices in CS, as generating these matrices is faster, and more importantly the matrix-vector multiplication process is faster when we use these matrices which in turn makes the reconstruction algorithm faster. Moreover, they arise in applications such as identifying  linear time-invariant systems \cite{bajwacompressed}.  We know that CS can be applied in MRI using partial Fourier matrices. However, application of CS in MRI with using  a generalization of circulant matrices called \textit{Toeplitz matrices} has been proposed in \cite{mri}, and has been shown that it has an advantage over the classical method in the sense that it spreads out the image energy more evenly.

\begin{defi}

A circulant matrix is a matrix of the form $$C=\begin{bmatrix} c_1 & c_2 & .... & c_n \\ c_n & c_1 & .... & c_{n-1} \\ \vdots \\ c_{2} & c_3 & ... & c_1 \end{bmatrix}$$

\noindent We say that $C$ is generated by the vector $v=(c_1, c_2, ... ,c_n)$.

\end{defi}

It is seen in the definition above that a circulant matrix is a square matrix and hence, is not appropriate to use in the context of CS. Therefore, in practice, random rows of such matrices have been considered to obtain a class of matrices called ``partial circulant matrices". It is shown in \cite{rauhutcirculant2} that if a circulant matrix is generated by a Rademacher sequence, then with high probability, the partial circulant matrix obtained from choosing $m$ rows of the (square) $n \times n$ circulant matrix satisfies RIP with $\delta_{2k}<1/\sqrt{2}$ if $m \gtrsim k^{3/2} \log^{3/2} n$. This condition was later improved by Krahmer et al. \cite{mendelson2} to the (suboptimal) condition $m \gtrsim k \log^2 k \cdot \log^2 n$. Hence, to compare these random matrices with the sub-Gaussian random matrices, we observe that there is an additional $\log^2 k$ factor in the expression for the minimum number of measurements. However, on the positive side, these matrices can be diagonalized using Fourier transform, and therefore, using FFT,  the process of matrix-vector multiplication as well as running the reconstruction algorithm become faster.

 Now, we give an explicit formula for the novel class of deterministic \textit{partial circulant} matrices.  We obtain bounds for the coherence of these matrices, which we then use to identify requirements that relate the sparsity level $k$ to the number of measurements $m$ (as well as the ambient dimension $n$) to ensure that these matrices can be used as CS matrices. We also perform some numerical experiments to compare the performance of these matrices with other deterministic and  random CS matrices.  Finally, we investigate the problem of quantization when these matrices are used as measurement matrices, and we provide theoretical guarantees for the error in reconstruction using one-stage recovery method in the case of $r$th order $\Sigma \Delta$ quantization.

\section{A novel, explicit construction}
\label{ourconst2}

Consider the $p \times p$ deterministic matrix $A$  defined by $$\begin{aligned} A_{i,j}=\begin{cases} \Big( \frac{j-i}{p} \Big)  & \mbox{if $i \ne j$} \\ 1 & \mbox{if $i=j$}  \end{cases} \end{aligned} $$ where $\Big( \frac{.}{p} \Big)$ denotes the Legendre symbol, and $1 \leq i,j \leq p $. \\

\begin{prop} Let $p$ be any prime number, then the matrix $A$ as defined above is a circulant matrix. 
\end{prop}

\begin{proof} First, we define the operator $\mathcal{S}: \mathbb{R}^p \to \mathbb{R}^p$ as follows. $$\mathcal{S} \Big( (x_1,x_2,...,x_p) \Big) := (x_p,x_1, 
\cdots ,x_{p-1})$$

\noindent Now, if we denote the rows of $A$ by $A^1$, $A^2$, ..., $A^{p}$, then to prove $A$ is a circulant matrix, we need to show $\mathcal{S}(A^i)=A^{i+1}$ for $1 \leq i \leq p-1$. Next, 

$$\begin{aligned} \mathcal{S}(A^i) & =\mathcal{S} (A_{i1},A_{i2},...,A_{ip}) =\mathcal{S} \Big( (\frac{1-i}{p} ), (\frac{2-i}{p}),...,(\frac{p-i}{p}) \Big) \\ & =\Big( (\frac{p-i)}{p}), (\frac{1-i}{p}),...,(\frac{p-1-i}{p}) \Big) \\ &  = \Big( (\frac{1-(i+1)}{p}), (\frac{2-(i+1)}{p}),...,(\frac{p-(i+1)}{p}) \Big) \\ & = (A_{i+1,1},A_{i+1,2},...,A_{i+1,p}) = A^{i+1} \end{aligned} $$

\noindent as desired. \end{proof}

\noindent Given the matrix $A$ above, we define the $\lceil p^{3/4} \rceil \times p$ (where $\lceil x \rceil$ denotes the ceiling of $x$) measurement matrix $\Phi$ by considering the first $\lceil p^{3/4} \rceil$ rows of the matrix $A$  (and then normalizing it).

\begin{defi} \label{partialdefi} Let $p \geq 3$ be a prime number.  The $(i,j)$ th entry of the measurement matrix $\Phi$, of our construction is defined as follows. $$\begin{aligned} \Phi_{i,j}=\begin{cases} \frac{1}{\sqrt{\lceil p^{3/4} \rceil }}\Big( \frac{j-i}{p} \Big)  & \mbox{if $i \ne j$} \\ \frac{1}{\sqrt{\lceil p^{3/4} \rceil}} & \mbox{if $i=j$}  \end{cases} \end{aligned}$$

\noindent with $1 \leq i \leq \lceil p^{3/4} \rceil$, and $1 \leq j \leq p$ \end{defi} Based on this definition, $\Phi$ is a deterministic partial circulant matrix, and hence, all the features and benefits of using a circulant measurement matrix, such as less space required to store the matrix, or fast matrix-vector multiplication, and fast reconstruction scheme can be applied to these matrices.

\begin{Theor}
\label{const}
There exists $p_0 \geq 23$ such that for $p \geq p_0$, the coherence $\mu$ of the $\lceil p^{3/4} \rceil \times p $ matrix $\Phi$ defined above satisfies $$\mu \leq  \frac{3 p^{1/2} \log p}{ p^{3/4}}=  \frac{3 \log p}{p^{1/4}} $$    Hence, for these matrices $$\delta_k \leq k \mu \leq 3 k \frac{\log p}{p^{1/4}} $$ Thus, in the context of CS, these matrices can be used for measurement and $\ell_1$ recovery  of vectors with the sparsity level of $k=\mathcal{O} (\frac{p^{1/4}}{\log p })$.
\end{Theor}

%
\noindent Note that if we use $\Phi$ as a measurement matrix, then  the maximum sparsity level $k$, and the number of measurements $m$ are related via $k=\mathcal{O}\Big( \frac{m^{1/3}}{\log n} \Big)$ since as we observed above, the number of measurements is of order $p^{3/4}$, and the the maximum sparsity level $k$ is of order $\frac{p^{1/4}}{\log p}$. This clearly compares unfavourably  to the random case; however note that our construction is \textit{explicit} and \textit{fast}. 

\begin{rk}

 Similar results will hold if we construct a $\lceil p^{\alpha} \rceil \times p$ CS matrix $\Phi$ with the similar definition (with $ 1/2 < \alpha <1$) given in section II.  Specifically, we define the $\lceil p^{\alpha}\rceil \times p$ matrix $\Phi$ as $$\Phi_{i,j}= \frac{1}{\sqrt{\lceil p^{\alpha}\rceil}} \Big( \frac{j-i}{p} \Big) $$ where $1 \leq i \leq \lceil p^{\alpha} \rceil$ and $1 \leq j \leq p$. Such matrix $\Phi$ can be used for the exact recovery of signals with the sparsity level $k$ satisfying $$k \leq \frac{p^{\alpha-1/2}}{18 \sqrt{2} \log p }=\mathcal{O} \Big( \frac{m^{\frac{\alpha-1/2}{\alpha}}}{\log n} \Big)$$ 
 
 \noindent where we used the fact that $m=\lceil p^{\alpha} \rceil$, and $n=p$. Hence, if we choose a larger value for $\alpha$ (close to 1 ), then the matrix $\Phi$ can be used for recovery of a  larger class of signals but we need more number of measurements. While if we choose a smaller value for $\alpha$ (close to 1/2), then the matrix $\Phi$ is closer to be an ideal CS matrix (where the number of measurements is much less than the ambient dimension) but on the other hand $\Phi$ can be used for recovery of a  smaller class of signals. Note that $0<\frac{\alpha-1/2}{\alpha}<1/2$ for $1/2<\alpha<1$, which is within boundaries of square-root barrier for deterministic matrices. 
\end{rk}

\noindent To prove Theorem \ref{const}, we should consider the inner product of two distinct columns of our measurement matrix. As we will see below, the inner products are related to a quantity called ``\textit{incomplete Weyl sums}". Let $\chi_p$ be a character modulus $p$, and let $f(x) \in \mathbb{Z}_p[x]$ be a polynomial. A \textit{complete Weyl sum} is a sum of the form $$\sum_{x=0}^{p-1} \chi_p \Big( f(x) \Big). $$ 

\noindent It can be shown \cite{residues} that if $f(x) \in \mathbb{Z}_p [x]$ is monic of degree $d \geq 1$ and $f$ has distinct roots in $\mathbb{Z}_p$, then $$\Big| \sum_{x=0}^{p-1} \chi_p \Big( f(x) \Big) \Big| <d \sqrt{p} $$

\noindent An incomplete Weyl sum is a sum of the form $$\sum_{x=M}^{N} \chi_p \Big( f(x) \Big), $$

\noindent where $M, N \in \mathbb{Z}_p$. In the case of $f(x)=x$, an upper bound was found by Polya in 1918 : $$ \Big| \sum_{x=M}^{N} \chi_p(x) \Big| \leq \sqrt{p} \log p $$ 

\noindent In the case of an arbitrary polynomial $f(x)$, the following bound on the Incomplete Weyl sums can be derived.

\begin{Theor}  \label{legendresymb} (Incomplete Weyl-Sum Estimate, Theorem 9.2 in \cite{residues}) : There exists $p_0>0$ such that for any prime number $p \geq p_0$, and for any monic polynomial $f(x) \in \mathbb{Z}_p$ of degree $d \geq 1$ and with distinct roots in $\mathbb{Z}_p$ we have $$ \Big| \sum_{x=0}^{N} \chi_p \Big( f(x) \Big) \Big| \leq d(1+\log p) \sqrt{p}$$

\noindent for any integer $N \in \mathbb{Z}_p$. 

\end{Theor}

\begin{proof}[Proof of Theorem \ref{const}]

Suppose that $\Phi_a$ and $\Phi_b$ are two distinct columns of $\Phi$, i.e., $a \ne b$. Then \begin{equation} \label{mu} \langle \Phi_a , \Phi_b \rangle = \frac{1}{\lceil p^{3/4} \rceil } \sum_{i=1}^{ \lceil p^{3/4} \rceil } \Big( \frac{b-i}{p} \Big) \Big( \frac{a-i}{p} \Big) \end{equation} Next, let $f(x):=(x-a)(x-b)$, $\chi (n) := \Big( \frac{n}{p} \Big)$ (which is a character modulus $p$), and $N=\lceil p^{3/4} \rceil $. Then, $f(x)$ is of degree $d=2$, and hence we can use Theorem \ref{legendresymb}.  Also, note that for any integer $i \in \mathbb{Z}_p$, we have $f(i)=(b-i)(a-i)$, and hence, $$\chi \Big( f(i) \Big) = \Big( \frac{(b-i)(a-i)}{p} \Big) =(\frac{b-i}{p} ) (\frac{a-i}{p} ) $$

\noindent Therefore, 

$$  \begin{aligned}   \Big| \sum_{i=1}^{\lceil p^{3/4} \rceil }    \Big( \frac{b-i}{p} \Big) \Big( \frac{a-i}{p} \Big) \Big| &  =   \Big| \sum_{i=1}^{\lceil p^{3/4}\rceil}  \chi_p  \Big( f(i) \Big) \Big| \\ & \leq \Big| \sum_{i=0}^{\lceil p^{3/4}\rceil}  \chi_p  \Big( f(i) \Big) \Big| +1 \\ & \leq 2(1+\log p) \sqrt{p} +1 \leq (3+2 \log p) \sqrt{p} \leq (3 \log p) \sqrt{p}  \end{aligned} $$

\noindent where we used the assumption $p \geq 23$ to conclude $\log p \geq 3$, i.e., $p \geq e^3$. \\

\noindent Since the inequality above is valid for \textit{any} two distinct values of $a$ and $b$,  using \eqref{mu} we can conclude $$ \mu = \max_{ a \ne b} \Big| \langle \Phi_a, \Phi_b \rangle \Big| \leq \frac{3 p^{1/2} \log p}{p^{3/4}}= \frac{3 \log p}{p^{1/4}}$$ as desired. 

\end{proof}

\section{One-stage recovery for $\Sigma \Delta$-quantized compressed sensing with deterministic partial circulant matrices}
\label{onestagenewconst}

When a sparse or compressible signal is acquired according to the principles of CS, the measurements still take on values in the continuum. In today's digital world a subsequent quantization step, where these measurements are replaced with the elements from a finite set is crucial.  In \cite{rongrong} an efficient approach for quantization in CS has been derived: $\Sigma \Delta$ quantization followed by a one-stage tractable reconstruction method. It is shown in \cite{rongrong} that this method is stable respect to noise and robust respect to compressible signals. However, there was one caveat in their work: it was applied only for the class of sub-Gaussian matrices. In \cite{armanozgur}, this method was generalized by abstracting out one main underlying property (called Property (P1)) as defined as follows.

\noindent \textbf{Property (P1).} Suppose that $\Phi$ is an $m \times n$ unnormalized CS measurement matrix, with (expected) column norm of $\sqrt{m}$. We say that $\Phi$ satisfies the property (P1) of order $(k,\ell)$ if the RIP constant of $\frac{1}{\sqrt{\ell}} (\Phi)_{\ell}$---where $(\Phi)_{\ell}$ is the restriction of $\Phi$ to its first $\ell$ rows---satisfies $\delta_{2k} < 1/9$. 
\smallskip

 It is shown in \cite{armanozgur} that we can perform a one-stage reconstruction $\Sigma \Delta$ quantization method for any CS measurement matrix as far as the measurement matrix $\Phi$ satisfies the property \textit{(P1)} of order $(k,\ell)$, as defined above.  In particular, it is shown that if $D$ is the difference matrix, and we write the singular value decomposition of $D^r$ in the form $D^r=U \Sigma V^T$, and if the (original) measurement matrix satisfies \textit{(P1)} of order $(k,\ell)$, then we can use a (modified) measurement matrix $\tilde{\Phi}=U \Phi$, and the error in reconstruction using the algorithm 
 
 \begin{align} (\hat{x},\hat{\nu})  := \arg \min_{(z,\nu)} \|z\|_1 & \mbox{ s.t. }   \|D^{-r} (U \Phi z +\nu-q) \|_2 \leq C_r \delta \sqrt{m},\notag \\ 
& \mbox{ and } \| \nu \|_2 \leq \epsilon \sqrt{m}. \label{onestagesol}
\end{align}  
  satisfies \begin{equation} \label{lm6} 
\|x-\hat{x}\|_2 \leq C \Big( (\frac{m}{\ell})^{-r+1/2}\delta +\frac{\sigma_k(x)}{\sqrt{k}} + \sqrt{\frac{m}{\ell}}\epsilon\Big) \end{equation}
where $C$ is a constant that does not depend on $m,\ell,n$.  Now, we show that our proposed construction satisfies the Property (P1).

\begin{prop} The partial circulant matrix $\Phi$, as defined in Definition \ref{partialdefi}, satisfies the property (P1)  of order $(k,\ell)$ for $\ell= \lceil p^{5/8} \rceil$, and $k=\lceil \frac{p^{1/8}}{\log p} \rceil$. Hence, in terms of number of measurements $m=\lceil p^{3/4} \rceil $, and the ambient dimension  $n=p$, we have  $\ell=\mathcal{O}(m^{5/6})$, and $k=\mathcal{O}(\frac{m^{1/6}}{\log n})$. 
\end{prop}

\begin{proof} In one stage recovery for $\Sigma \Delta$-quantized CS as described in \cite{rongrong, armanozgur}, we start with a matrix that satisfies RIP. Then, we consider the first $\ell$ (for $\ell$ as small as possible) rows of the $m \times n$ matrix $\Phi$, and the resultant matrix still satisfies RIP. The technical issue is that such property does not hold for an arbitrary measurement matrix, e.g., see Section IV of \cite{armanozgur}.  However, our construction here has the advantage that the matrix obtained by considering its first $\ell$ rows is still a partial circulant binary matrix. Hence, to obtain the RIP constants of a matrix of the form \begin{equation} \label{ellmatrix} \Phi=\frac{1}{\sqrt{\ell}} \Phi_{\ell} \end{equation}

\noindent where $\Phi_{\ell}$ is the restriction of $\Phi$ as defined above to its first rows, we simply replace $m=\lceil p^{3/4} \rceil $ in the proof of Theorem \ref{const}, with $\ell=  \lceil p^{5/8} \rceil $. Then, we conclude that the coherence of this matrix given in \eqref{ellmatrix}, satisfies $$\mu \leq \frac{18 \sqrt{p} \log p}{p^{5/8}}= \frac{18 \log p}{p^{1/8}}$$ 

\noindent Hence, if the sparsity level satisfies $k \leq \frac{p^{1/8}}{\log p}$, we have $$\delta_k <k\mu=\frac{1}{9}$$

\end{proof}

 \noindent The following Corollary is immediate by combining the theorem above, and Theorem \ref{recoveryp1}. 
 
 \begin{cor}
 
 Consider the deterministic partial circulant matrix $\Phi$ of order $\lceil p^{3/4} \rceil \times p$, as defined in Definition \ref{partialdefi}. Then, for $k\leq \frac{p^{1/8}}{\log p}$, any $k$-sparse signal $x$ can be approximated with  the vector $\hat{x}$, the solution to \eqref{onestagesol}. Here, the measurement matrix that we use is $U \Phi$,  and  $q$ is obtained from $rth$ order $\Sigma\Delta$ quantization. Moreover, as we increase the number of measurements $p$, the error in approximation decreases according to \begin{equation} \label{kfixed} \|x-\hat{x} \|_2 \leq  2 C_r C_4 (3 \pi r)^r (m^{1/6})^{-r+1/2} +2C_4 \epsilon \sqrt{m^{1/6}}  \end{equation} for constant $C_4$ depending only on RIP constant of $\Phi$, and where $m=\lceil p^{3/4} \rceil$ denotes the number of measurements. 

\end{cor} 

\section{Numerical experiments}

First, we give an example of the new construction when $p=997$.  Figure \ref{matrix-figure3} shows the matrix constructed with this method. 

%
%
%
%

To compare the construction we introduced in this chapter with other existing deterministic constructions extensively, we run four numerical experiments. 

 In the first experiment, we compute the coherence of  matrices as defined in Definition \ref{partialdefi}, in order to find out about the tightness of the theoretical bound we derived in this chapter. In this experiment, for each value of $p$, with $71 \leq p \leq 1193$,  we start with the $\lceil p^{3/4} \rceil \times p $ matrix $\Phi$, then we construct  $\Phi^* \Phi$. Next, we calculate the maximum off-diagonal entry (in absolute value) of this Gramian matrix.  This gives the coherence of the matrix $\Phi$ and we plot this in log-log scale in Figure \ref{coherence4}. As seen  in this Figure, the coherence of these matrices as a function of $p$, behaves almost like $f(p)=\frac{1}{p^{1/4}}$, which apart from a logarithm factor, matches with the bound we derived in Theorem \ref{const}. In order to make a comparison, we also plot the coherence of chirp sensing matrices (of the size $p \times p^2$) as a function of $p$. Note that the coherence of chirp sensing matrices matches precisely with $f(p)=p^{-1/2}$. 

In the second  experiment, we fix a measurements matrix $\Phi$ (random or deterministic) and we consider a $k$-sparse signal $x \in \mathbb{R}^{300}$. For each $2 \leq k \leq 102$, we choose a random support $T \subseteq \{1,2,...,300\}$ with $k$ elements, and we choose non-zero entries of $x$ independently from standard Gaussian distribution.  Then, we compute the measurement vector $y=\Phi x$, and we use \textit{BP} to  to find the reconstruction vector $\hat{x}$. Next, we compute SNR for the signal $x$ defined by \begin{equation} \label{snrd} \mbox{SNR} (x) = 10 \cdot \log_{10} \Big( \frac{ \|x \|_2 }{\|x-\hat{x} \|_2} \Big) dB \end{equation} 

\noindent If we obtain $\mbox{SNR}(x) > 50$, we consider the reconstruction as a \textit{successful} recovery and otherwise an \textit{unsuccessful} recovery. For each $k$, we repeat this experiment 10 times and we  let $f(k) := \frac{\mbox{number of successful recoveries}}{10}$. In Figure \ref{constructionlegendre}, we plot $f(k)$ vs. $k$ for different constructions. For our proposed construction, we choose a prime number close to 1000, say $p=997$, which leads to the number of measurements $m=\lceil 997^{3/4} \rceil=178$. We can choose the same number of measurements for random Bernoulli matrix. For Reed-Muller construction, since the number of measurements is a power of 2, we choose the smallest power of 2, greater than 178, namely, $m=256$. For DeVore construction, we choose $m=169$ (the closest integer of the form $p^2$ to 178). Accordingly, in Figure \ref{constructionlegendre}, we use Reed-Muller construction of the size $256 \times 2^{36}$, restricted to its first 300 columns (as $x\in \mathbb{R}^{300}$), DeVore construction of the size $13^2 \times 13^3$, restricted to its first 300  columns, the novel construction of the size $\lceil 997^{3/4} \rceil \times 997=178 \times 997$, restricted to its first 300 columns, and also random Bernoulli of the size $178 \times 300$. As we see in this Figure, the novel construction introduced in this chapter has the best performance among the other deterministic constructions mentioned above, and its performance is comparable to the random Bernoulli matrices.


 In the third experiment, we fix $k=10$, and we consider a $k$-sparse signal $x \in \mathbb{R}^{40}$, with random support and non-zero entries chosen independently from the standard Gaussian distribution, and we consider the new construction with $m=\lfloor p^{3/4} \rfloor $, and $n=p$ (with $41 \leq p \leq 293$). For each value of $p$, we evaluate $y=\Phi x$, we approximate $x$ with $\hat{x}$ using \textit{BP}, and we evaluate SNR using \eqref{snrd}. Similar to above, if $SNR > 50$, we consider the reconstruction as a successful recovery, and otherwise an unsuccessful recovery. We repeat the same experiment 10 times and we define $f(p) :=\frac{\mbox{number of successful recoveries}}{10}$. We plot $f(p)$ vs. $p$ in Figure \ref{matrix-figure}. We repeat the same process above with $k=20$. Lastly, we perform the same process with $k=10$ and $k=20$ using random Bernoulli matrices with exactly same sizes, i.e., $\lfloor p^{3/4} \rfloor \times p$, with $41 \leq p \leq 293$. Figure \ref{matrix-figure} also suggests that the performance of the new construction is comparable with the random Bernoulli matrices.


\begin{figure}

    \centering
    \includegraphics[  width=0.4\textwidth   ]{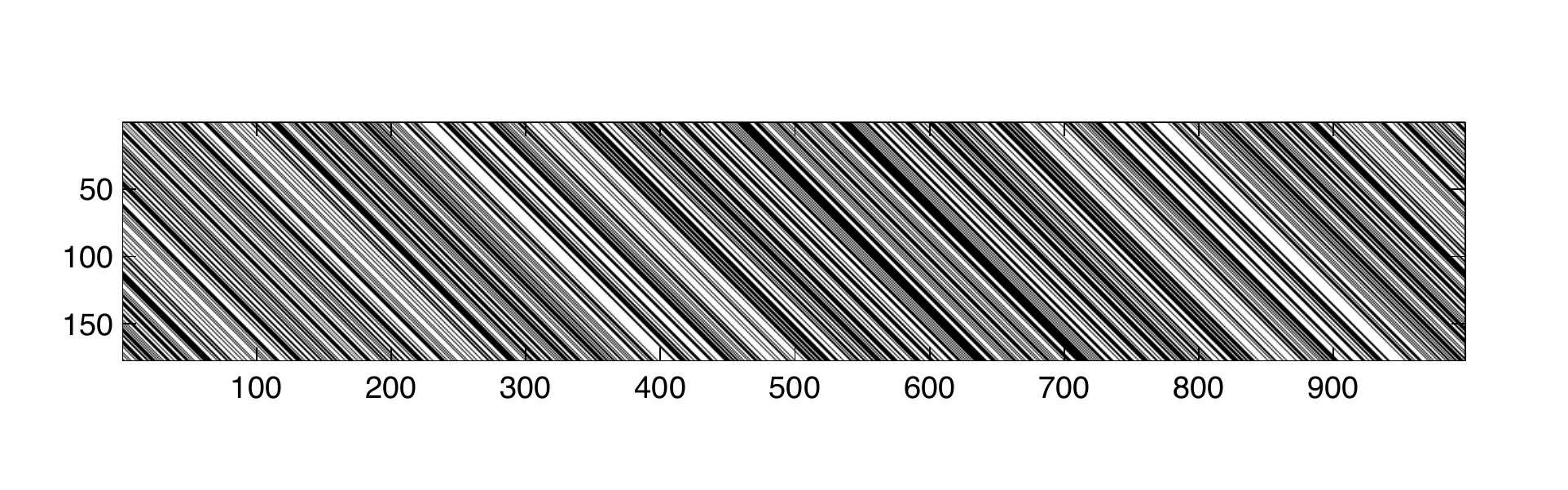}

\caption{ The binary matrix given by the new construction for $p=997$. }
\label{matrix-figure3}
\end{figure}

\begin{figure}

    \centering
    \includegraphics[ width=0.4\textwidth]{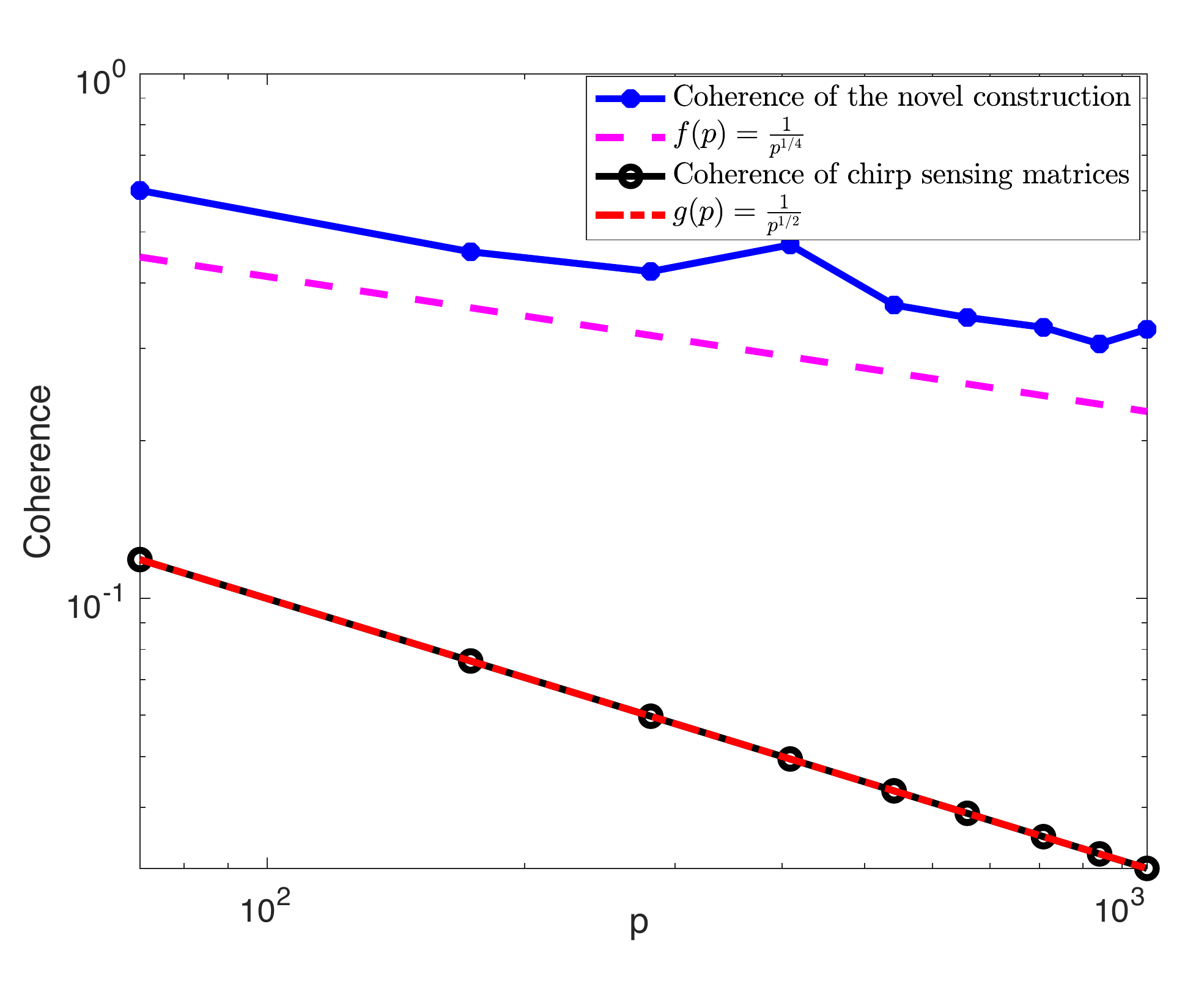}

\caption{ Coherence of the matrices introduced in this chapter, and also chirp sensing matrices in log-log scale, accompanied with best fitted lines. As seen in this Figure, the coherence of our construction behaves as $\sim m^{-1/3}$,  and for chirp sensing matrices behaves as $\sim p^{-1/2}=m^{-1/2}$.  }
\label{coherence4}
\end{figure}

\begin{figure}

    \centering
    \includegraphics[ width=0.4\textwidth ]{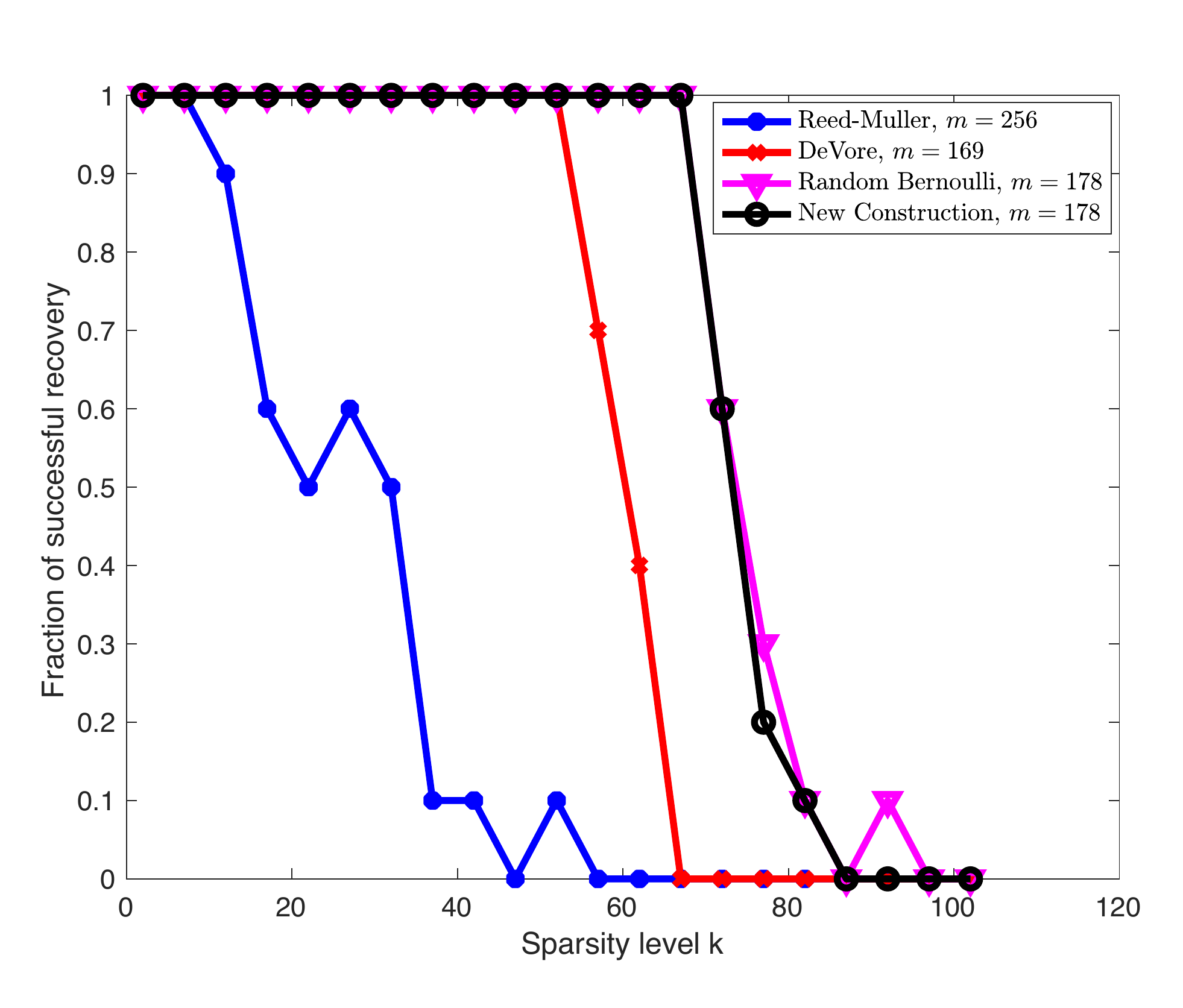}

\caption{ The fraction of exactly recovered vectors versus sparsity for a fixed number of measurements. The number of measurements is chosen as $m=256$ for the Reed-Muller matrix, $m=178$ for the new construction and also for the Bernoulli matrix, and $m=169$ for the DeVore's construction. The ambient dimension of all signals is $n=300$ }
\label{constructionlegendre}
\end{figure}

\begin{figure}

    \centering
    \includegraphics[width=0.4\textwidth ]{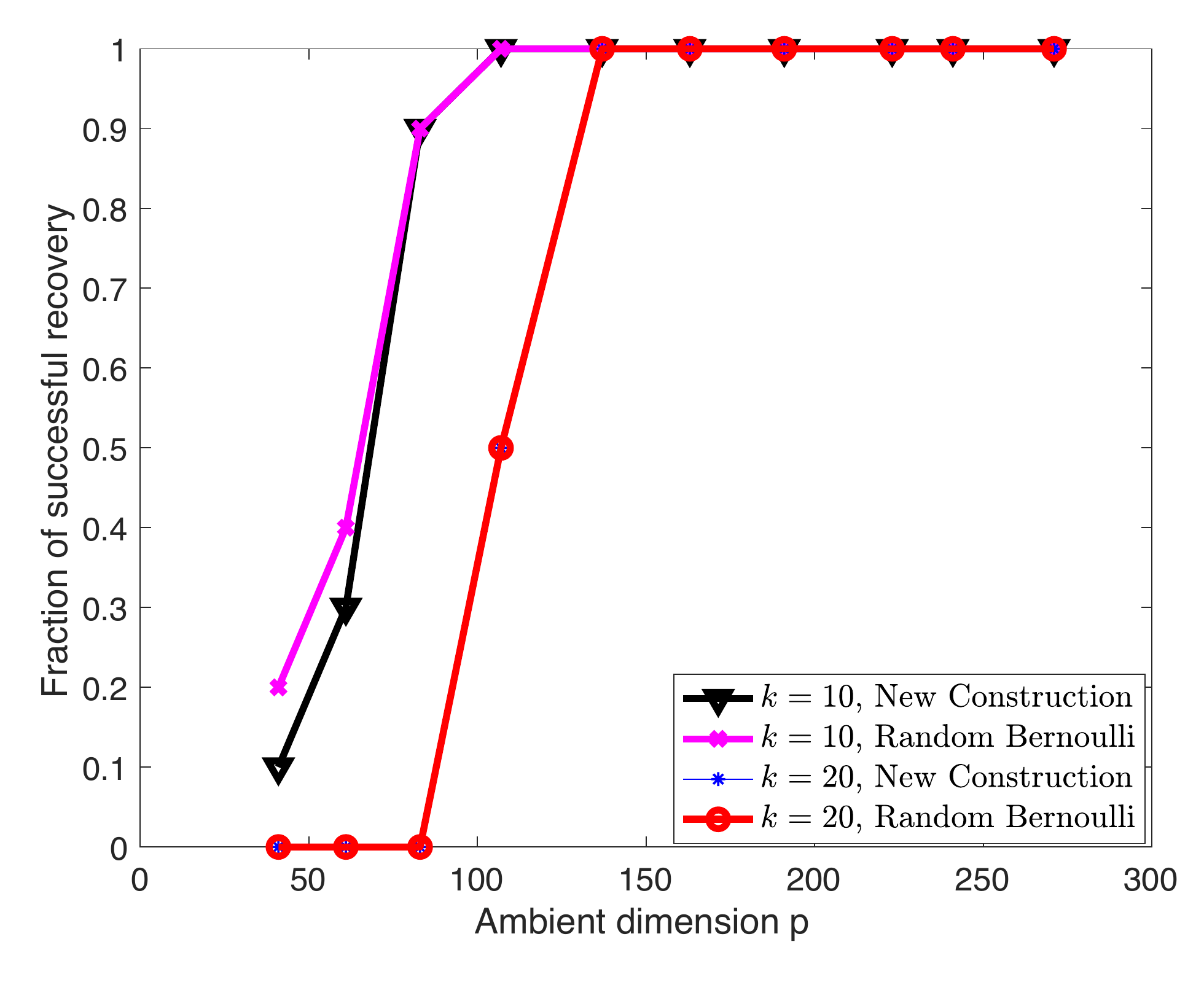}

\caption{ The graph of fraction of exactly recovered vectors (for 10 experiments) versus prime number $p$ for a fixed level of sparsity ($k=10$ or 20) for the new construction and the Bernoulli matrices. Note that only three graphs are shown because the graphs corresponding to $k=20$ for the new construction and random Bernoulli exactly overlap with each other. This suggests that our proposed deterministic construction has a very similar performance to random Bernoulli.}
\label{matrix-figure}
\end{figure}

\begin{figure}

    \centering
    \includegraphics[width=0.4\textwidth ]{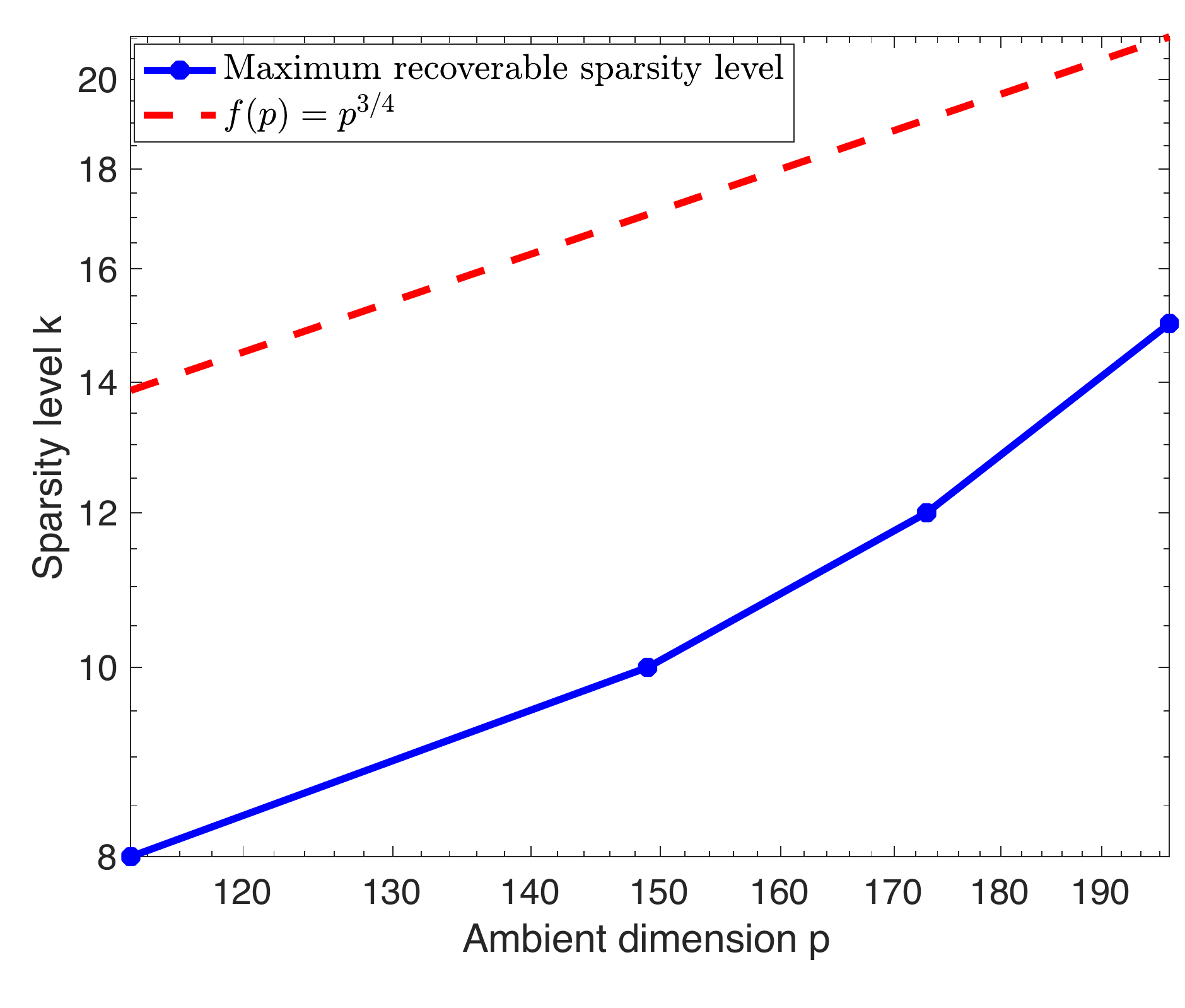}

\caption{ The maximum sparsity level of  recoverable signals $g(p)$ versus the prime number $p$ compared with the graph of $f(p)=p^{3/4}$.}
\label{compsubopt}
\end{figure}



In the last experiment, for a fixed value of $p$ (which fixes the number of measurements $m=\lceil p^{3/4} \rceil$), and with $113\leq p \leq 197$, we consider a 1-sparse signal $x \in \mathbb{R}^{100}$ with random support and the non-zero entry chosen from a standard Gaussian distribution. Then, we  reconstruct it using \textit{BP}, and we measure $\mbox{SNR}(x)$. We repeat the experiment 50 times for 50 signals $x$ and if the minimum value  of $\mbox{SNR}(x)$ becomes greater than 50 dB, we consider that level of sparsity as exactly recoverable and we increase the sparsity level by 1 unit. Then, we repeat the same experiment until the minimum value of $\mbox{SNR}(x)$ (among all 50 experiments) becomes less than 50 dB for a sparsity level $k_1$. Then we define $g(p):=k_1-1$ as the maximum level of recoverable sparsity corresponding to $p$. The graph of $g(p)$ versus $p$ is plotted in Figure \ref{compsubopt} in log-log scale.  In the suboptimal case, we know that the maximum level of sparsity is $m/ \log(n/m)$ which will be $p^{3/4}/ (4 \log p)$. Therefore, in addition to the graph of  $g(p)$, we plot the graph of $f(p)=p^{3/4}$ and we compare these graphs. As it can be seen in this Figure, our construction, like the random CS matrices has the feature that numerically, the maximum recoverable sparsity level behaves close to the suboptimal case.

\bibliographystyle{plain}
\bibliography{novel_const_01}



%
%


\end{document}